\documentclass[a4paper]{article}
\usepackage{amsmath,amsthm,amssymb, mathrsfs}

\usepackage{fullpage}
\usepackage{enumerate}
\usepackage{setspace}
\usepackage{listings}
\usepackage{tikz-cd}
\usepackage{tikz}
\usepackage{tocbibind}
\usepackage{comment}
\usepackage{mathtools}
\usepackage[english]{babel}
\usepackage[autostyle]{csquotes}
\usepackage[
  backend=biber,
style=ieee]{biblatex}
\usepackage{hyperref}
\newcommand{\Co}{\mathbb{C}}
\newcommand{\Z}{\mathbb{Z}}
\newcommand{\N}{\mathbb{N}}

\def \r {\rightarrow}

\def \lan {\langle}
\def \ran {\rangle}

\def \meet {\wedge}
\def \join {\vee}
\def \De {\Delta}

\newtheorem{theorem}{Theorem}[subsection]
\newtheorem{defin}[theorem]{Definition}
\newtheorem{lemma}[theorem]{Lemma}

\newtheorem{exmp}[theorem]{Example}

\newtheorem{prop}[theorem]{Proposition}

\title{An Error Correctable Implication Algebra for a System of Qubits}
\author{Morrison Turnansky\thanks{  mturnansky@gmail.com}}

\begin{document}

\maketitle
\begin{abstract}
  We present the Lukasiewicz logic as a viable system for an
  implication algebra on a system of qubits. Our results show that
  the three valued Lukasiewicz logic can be embedded in the
  stabilized space of an arbitrary quantum error correcting stabilizer code. We then fully
  characterize the non trivial errors that may occur up to group
  isomorphism. Lastly, we demonstrate by explicit algorithmic example,
  how any algorithm consistent with the Lukasiewicz logic can
  immediately run on a quantum system and utilize the indeterminate state.
  \newline

  \noindent ACM-class: F.4.1
\end{abstract}
\newpage
\section{The Measurement of Pauli Gates and MV Algebra}
Modern computation and boolean algebras are intrinsically linked. The
behavior of trillions of transistors is completely described and
represented by a
corresponding boolean algebra due to the strict property that each
transistor can be in exactly one of two states. As we move forward to quantum
computation, we gain an additional indeterminate state. In this
section, we present a number of structures for three
valued logics and algebras and highlight their interrelationships.
We demonstrate that a Lukasiewicz logic relates to an algebraic qubit
system in an equivalent way to how
a Boolean logic relates to $\{0,1\}$. Finally, we will give a
physical realization of the Lukasiewicz logic as a series of quantum gates.

\subsection{Introduction to three valued logics}
When considering a starting point for a 3 valued algebraic system,
there are a number of choices.  We start with the more general structure
and progressively increase specialization to support our desired applications.

\begin{defin}[\cite{mv-algebra-def}]
  An MV-algebra is an algebraic structure $\lan A, \oplus, \neg, 0
  \ran$ consisting of a nonempty set A, a binary operation $\oplus: A
  \times A \r A$, a unary operation $\neg: A \r A$, and a constant
  $0$ denoting a fixed element of $A$, which satisfies the following identities:
  \begin{enumerate}
    \item $(x \oplus y) \oplus z = x \oplus (y \oplus z)$
    \item $x \oplus 0 = x$
    \item $x \oplus y = y \oplus x$
    \item $\neg \neg x = x$
    \item $x \oplus \neg 0 = \neg 0$
    \item $\neg(\neg x \oplus y) \oplus y = \neg(\neg y \oplus x ) \oplus x$.
  \end{enumerate}
\end{defin}

\begin{defin}
  Let $A = \{0, \frac{1}{2}, 1\}$,  $\oplus A \times A \r A$ defined
  by $(x,y) \mapsto min(x+y, 1)$, and $^\neg:A \r A$ defined by $x
  \mapsto  1 - x$. The standard  $MV_3$ algebra is defined as
  $\mathscr{A}  = (A, \oplus, \neg)$.
\end{defin}

\noindent We can now define a Lukasiewicz logic both as an independent
structure and as
a specific instance of a $MV_3$ algebra.

\begin{defin}[\cite{logic_many_valued}]
  Lukasiewicz logics, denoted by $L_m$ in the finite case and
  $L_\infty$ in the infinite case, are defined by the logical matrix,
  which has either some finite set $W_m = \{\frac{k}{m-1} |  0 \le k
  \le m - 1 \}$ of rationals within the real unit interval or the
  whole interval $W_\infty = [0,1] = \{x \in \mathbb{R} | 0 \le x \le
  1\}$ as the truth degree set. The degree 1 is the only designated
  truth degree. The main connectives of these systems are a strong
  and weak conjunctive, $\&$ and $\meet$, respectively, given by the
  truth degree functions
  \begin{align*}
    u \, \& \, v &= max\{0 , u + v - 1\}\\
    u \meet v &= min\{u, v\},
  \end{align*}
  a negation connective $^\neg$ determined by
  \[\neg u = 1 - u,\]
  and an implicative connective $\r$ with the truth degree function
  \[u \r v = min\{1 , 1 - u + v\}.\]
\end{defin}

\begin{prop}
  The standard $MV_3$ algebra, $\mathscr{A}$, with the addition of
  the axiom $x \oplus x \oplus x = x \oplus x$, defines the three
  valued Lukasiewicz logic.
\end{prop}

\begin{proof}
  By definition, we see that the states of $\mathscr{A}$ are exactly
  the states of $L_3$. Furthermore, negation of $L_3$ is precisely
  the same as negation in $\mathscr{A}$, so we only need to show
  implication may be defined in terms of $\oplus$ and $\neg$.

  Now we claim that $\neg x \oplus y = x \r y$:
  \begin{align*}
    \neg x \oplus y &= min(1, \neg x + y) \\
    &= min(1, (1-x) + y)\\
    &= min(1, 1-x+ y)\\
    &= x \r y.
  \end{align*}
\end{proof}

\begin{prop}[\cite{paraconsistency}]
  As a direct derivation of the axioms of a Lukasiewicz logic and the
  definition of
  $\rightarrow$, we include the following equations for completeness.
  \begin{align*}
    0 &= \neg 1\\
    x^2 &= x \odot x\\
    x \ominus y &= x \odot \neg y\\
    x \oplus y &= \neg (\neg x \odot \neg y)\\
    x \rightarrow y &= \neg x \oplus y\\
    | x - y | &= (x \ominus y) \oplus (y \ominus x)\\
    2x &= x \oplus x\\
    x \leftrightarrow y &= (x \rightarrow y) \odot (y \rightarrow x).
  \end{align*}
\end{prop}

\subsection{The Measurement of Pauli Gates and MV Algebra}
On the surface, the axioms of the cubic algebra look unrelated to the
axioms of the $MV_3$ algebra,
but in fact, $MV_3$ algebras completely
describe finite cubic algebras. In
\cite{turnanskypauli}, they show
how to embed the cubic lattice as a set
of quantum gates. We will use these results to lift a more standard logic,
the three valued Lukasiewicz logic, to a realization as quantum gates. We
hope that this will lay the foundation to a more general
theory of quantum computation.

For ease of the reader, we include the definition of a cubic lattice.

\begin{defin}[\cite{oliveira}]\label{cubic lattice definition}
  A cubic lattice, $C$, is a lattice with 0 and 1 satisfies the
  following axioms:
  \begin{enumerate}
    \item For $x \in L$, there is an order-preserving map $\Delta_x:
      (x) \r (x)$, $(x)$ denotes the principal ideal generated by $(x)$.
    \item If $0 < a,b <  x,$ then $a \join \Delta_x(b) < x$ if and
      only if $a \meet b = 0$.
    \item $L$ is complete.
    \item $L$ is atomistic.
    \item $L$ is coatomistic.
  \end{enumerate}
\end{defin}

As a cubic lattice is a lattice, it has two binary operations,
$\meet$, $\join$, and by the above definition, a third binary
operation $\Delta$. In order to relate the cubic lattice to a $MV_3$
algebra, we have the following operations.

\begin{defin}[\cite{paraconsistency}]
  For all $x,y \in M$,
  \begin{align*}
    x \join y &= (y \odot (x \leftrightarrow y) ^2 ) \oplus (\epsilon
    \odot (\epsilon \leftrightarrow | x- y|)^2 ).
  \end{align*}
  For all $x,y \in M$ such that $|x - y | \rightarrow \epsilon = 1$,
  \begin{align*}
    x \meet y &= (y \odot (x \leftrightarrow y ) ^2) \oplus (x \odot
    (x \leftrightarrow 2x)^2).
  \end{align*}
  For all $x,y \in M$ such that $y \join x = x$,
  \begin{align*}
    \Delta(x,y) &= (y \odot (x \leftrightarrow y)^2) \oplus (\neg y
    \odot (x \leftrightarrow \epsilon)^2).
  \end{align*}
\end{defin}

\begin{theorem}[Theorem 7 \cite{logic_many_valued}] \label{MV algebra as cubic}
  Let $M$ be a finite $MV_3$ algebra with an element $\epsilon = \neg
  \epsilon$. Then $\epsilon$ is unique, and the structure $(M, \join,
  \meet, \Delta)$ is an $n$-cubic algebra, where $n$ is the number of
  homomorphisms of $M$ into $L_3$. Further, all finite cubic algebras
  arise from this construction, and non-isomorphic finite $MV_3$
  algebras with self-negated elements determine non-isomorphic finite
  cubic algebras.
\end{theorem}

\begin{exmp}
  The mapping, $\phi$ from an $n$ cubic lattice to $MV_3$ algebra is
  as follows. For all $(A^+, A^-) \in CL$, $\phi((A^+, A^-))_i
  \mapsto 1$ if $i \in A^+$, $-1$ if $i \in A^-$, or $\frac{1}{2}$
  otherwise. We see that $\frac{1}{2}$ takes the place of our
  indeterminant element $X$ in our standard representation of the
  $MV_3$ algebra.
  Then $\epsilon$ is equal to the face representing the whole cube.
\end{exmp}
By our previous work in \cite{turnanskypauli}, we have an embedding
of the cubic lattice as quantum gates. Specifically, these
projections are onto subspaces of the Pauli gates, and by unitary
similarity, we can achieve any such eigenspace. These unitary
transformations are a special subset, namely tensor products of the
Pauli gates themselves.

\begin{theorem}\label{Lukasiewicz embedding}
  Let $H$ be a Hilbert space constructed as a tensor product of $2$
  dimensional spaces over a finite index set $I$. For the given
  Hilbert lattice $HL$ of $H$, there exists $MV_3$ algebra, $A$ such
  that $A \subseteq HL$, and the atoms of $CL$ are projections onto
  subspaces $H$ forming an orthonormal basis of $H$.

  Furthermore, the semantics of $MV_3$ algebra agree with the three
  valued Lukasiewicz logic.
\end{theorem}
\begin{proof}
  By Theorem 2.1.11  in \cite{turnanskypauli}, we have the claim that
  $CL \subseteq HL$. Now by \ref{MV algebra as cubic}, we have that
  $CL$ is a $MV_3$ algebra satisfying the semantics of the three
  valued Lukasiewicz logic.
\end{proof}

We have now shown that there is an embedding of the standard
representation of the $MV_3$ algebra as a set of quantum
operations. Specifically, we can represent them within the standard
set of Pauli gates. As the standard representation of the $MV_3$
algebra satisfies the Lukasiewicz three valued logic, we have an
embedding of the logic as well. Therefore, any semantics that have
been derived for these well studied algebras are representable as a
set of quantum gates.

\section{Applications of Cubic Lattices to Stabilizer Codes}
We have shown how the three valued Lukasiewicz logic can be
represented as standard set of quantum gates. A large problem in
quantum computing is error correction. We give an exposition of the
standard family of quantum error correction codes, namely quantum
stabilizer codes, and we then show that our embedding of the
three valued Lukasiewicz logic is well behaved in this system.

\subsection{Lukasiewicz logics and Pauli Groups}
The core principle of stabilizer codes is that our computational
state of interest lies in the common eigenspace of a set of
operators. For quantum stabilizer codes, the following operators serve
as  building blocks.
\begin{defin}
  The Pauli spin matrices, $\mathscr{S}$, are hermitian symmetries in
  $M_2(\Co)$ denoted by
  \[
    X =
    \begin{bmatrix}
      0 & 1 \\
      1 & 0
    \end{bmatrix}, \,
    Y =
    \begin{bmatrix}
      0 & -i \\
      i & 0
    \end{bmatrix}, \,
    Z =
    \begin{bmatrix}
      1 & 0 \\
      0 & -1
  \end{bmatrix}.\]
\end{defin}

In \cite{turnansky_thesis}, we constructed a Hilbert lattice and
showed that the $\De$ of Definition \ref{cubic lattice definition}
can be satisfied by an $|I|$ tensor product of $X$, and then by
unitary similarity, an $|I|$ tensor product of Pauli matrices, where
$I$ is an arbitrarily large, not necessarily finite, indexing set. We
should also highlight that we used a non-standard infinite tensor
product specified in \cite[Section 3]{turnansky_thesis}. We now adapt
our notation to the standard terminology of the quantum computing community.

When $I$ is a finite indexing set, the tensor product of Pauli
matrices constitutes a group of fundamental importance to quantum
stabilizer codes, the Pauli Group.

\begin{defin}
  The Pauli group on 1 qubit is the multiplicative group $G_1 = \{\pm
  I_2, \pm X, \pm Y, \pm Z\}$.
\end{defin}

\begin{defin}
  For some future applications, we often want to exclude the negative
  elements of Pauli spin matrices and the identity, so we define
  $G_1^+ = \{I_2, X, Y, Z\}$.
\end{defin}

\begin{defin}
  The Pauli group on $n \in \N$ qubits, $G_n$, is the group
  consisting of $n$-fold tensor products of elements of $G_1$.
  Similarly, we define $G_n^+$ as the group consisting of $n$-fold
  tensor products of elements of $G_1^+$.
\end{defin}

\begin{defin}
  We define the Pauli group on $|I|$ qubits, $G_{|I|}$ to be the
  multiplicative group of $|I|$ tensor products of elements in $G_1$.
  If the number of qubits is arbitrary, we will simply say the Pauli group.
\end{defin}

Here multiplication is computed in the typical way on each index $i
\in I$. We are now in a position to consider quantum stabilizer
codes. For the following definitions, the Pauli group acts on
$\Co^{2^{|I|}}$ by left multiplication. As the Pauli group is made up
of self adjoint operators, we can also view this action as the
observables of states.

\begin{defin}
  Let $\rho: G_n \r M_{2^n}(\Co)$ denote the unitary
  representation of the Pauli group.
\end{defin}

We now relate the quantum computing concepts to Lukasiewicz algebras.
Generally, these results follow from our prior work with cubic
lattices and the results of section 1.
\begin{defin}
  Let $G$ be a set of symmetries, then $Proj(G) = \{\frac{I \pm g}{2}
  | g \in G\} $  denotes the canonical projections of $G$.
\end{defin}

\begin{lemma}\label{projections form cubic lattice}
  Let $S$ be a stabilizing set generated by $\{s_i\}_{i = 1 }^n$,
  $s_i \in G_n$, $i = 1, \dots, n$. The meet semi lattice generated
  by the the canonical projections of the generators of $S$ satisfy
  the three valued Lukasiewicz logic as in Theorem \ref{Lukasiewicz embedding}.
\end{lemma}

\begin{proof}
  We begin by constructing a bijection between two sets:
  the co-atoms of the meet semi lattice of projections
  $Proj(\{s_i\}_{i=1}^n)$, and a set C of co-atoms of an $I$ cubic
  lattice, represented as a signed set.  We define a map $f:
  Proj(\{s_i\}_{i=1}^n) \r  C$ by $p \mapsto (i , \emptyset)$ if $p =
  \frac{1 + s_i}{2}$ or $p \mapsto (\emptyset, i)$ if $p = \frac{1 -
  s_i}{2}$. By construction, $f$ is injective. Since
  $|Proj(\{s_i\}_{i=1}^n)| = 2n = |C|$, $f$ is onto as well.

  Furthermore, $f$ is a meet semi lattice homomorphism by
  \cite{turnanskypauli}[Theorem 2.1.17], so we have
  that their sets of atoms are equal as well. Lastly, by Theorem
  \ref{MV algebra as cubic}, we conclude our result.
\end{proof}

We highlight that the above order isomorphism is only a meet
homomorphism and not a join homomorphism. This is a not a
contradiction of standard results as we are only claiming an order
isomorphism on a subset of the lattice generated by $S$.

Our notation for $S$ was purposeful. The Pauli Group is a fundamental
building block of quantum stabilizer codes.

\begin{defin}
  We say that a subspace $V_S \subseteq \Co^{2^{|I|}}$ is stabilized
  by $S \le G_{|I|}$ if for all $v \in V$ and all $s \in S$, $sv = v$.
\end{defin}

\begin{defin}
  An $[n,k]$ stabilizer code is defined to be the vector space $V_S$
  stabilized by the maximal subgroup $S$ of $G_n$ such that $-I \notin S$ and
  has $n - k$ independent and commuting generators $S = \{g_1, \dots,
  g_{n-k}\}$, $g_i \in G_{|I|}$, $i = 1, \dots n- k$. We denote this
  code by $C(S)$. \cite{nielsen}
\end{defin}

\begin{defin}
  Let $S$ be a stabilizing set generated by $\{s_i\}_{i = 1 }^k$,
  $s_i \in G_n$, $i = 1, \dots, k \le n$. We define a completion of
  $S$, to be a set a set of distinct generators $T = s_i$, $i = k+1, \dots n$,
  such that $\overline{S} = S \cup T$ is an abelian group.
\end{defin}

\begin{prop}\label{completion of S}
  Let $S$ be a stabilizing set generated by $\{s_i\}_{i = 1 }^k$,
  $s_i \in G_n$, $i = 1, \dots, k \le n$, then there exists a
  completion of $S$, $T$, and $T$ is a minimal set of generators
  projecting onto a basis $\mathbb{C}^{2^n}$.
\end{prop}

\begin{proof}
  By the standard representation for a stabilizer code, this is equivalent
  to completing a basis of vectors for the vector space $\Z_2^{2n}$ given $k$
  independent vectors.
\end{proof}

We now have a main result: quantum stabilizer codes and three
valued Lukasiewicz logics as different incarnations of the same structure.

\begin{lemma}\label{cubic lattice in stabilizer code}
  For any stabilizing set $S$ generated by $\{s_i\}_{i = 1 }^k$, $s_i
  \in G_n$, $i = 1, \dots, k \le n$, the canonical projections can be
  realized as a 3 valued Lukasiewicz logic.
\end{lemma}

\begin{proof}
  Let $T$ be the completion of $S$, which exists by Proposition \ref{completion
  of S}. Then by \ref{projections form cubic lattice}, we have a meet
  semi lattice isomorphism to a three valued Lukasiewicz logic.
\end{proof}

\begin{theorem}
  Let  $\mathscr{S}$ be the set of stabilizer codes on $n$ qubits and
  $\mathscr{L}$
  be the set of 3 valued Lukasiewicz algebras on $n$ states.
  There exists a well defined surjective order homomorphism $\Pi:
  \mathscr{S} \rightarrow \mathscr{L}$.
\end{theorem}

\begin{proof}
  For each stabilizer set $S \in \mathscr{S}$ the meet semi
  lattice generated by $Proj(S)$ forms a cubic lattice and therefore a
  three valued Lukasiewicz logic, $L$, so we define the map $\Pi:
  \mathscr{S} \r \mathscr{L}$ where $S \mapsto L$ by the construction in
  Theorem \ref{Lukasiewicz embedding}. For a given $S \in \mathscr{S}$,
  suppose $L_1, L_2 \in \mathscr{L}$ and $L_1$, $L_2$, both satisfy the
  semantics of $S$. Then $|L_1| = |L_2|$, as we have a fixed domain for
  finite valued Lukasiewicz logics, so $L_1 = L_2$. Lastly, suppose
  without loss of generality that $L \in \mathscr{L}$ has a set
  of co-atoms of cardinality $k$. There exists $S \in \mathscr{S}$
  also with a set of co-atoms of cardinality $k$, so that $\Pi(S) = L$.
\end{proof}




We now see that a cubic lattice and its associated three valued
Lukasiewicz logic make a strong candidate for a logic that can
semantically represent observables of a quantum stabilizer code. Importantly,
any stabilizer set
corresponds to a well defined Lukasiewicz logic, and, conversely, any finite
Lukasiewicz logic can be represented as quantum system generated by the
Pauli group. A primary application of the Pauli group in quantum
computing is their usage in error correction codes. In the
following section, we give an overview of this relationship, and then we
adapt our prior work with cubic lattices to this space while
demonstrating some new results.

\subsection{Automorphisms of the Cubic Lattice and Errors of Quantum
Stabilizer Codes}
Two of the primary focuses of \cite{turnansky_thesis} where embedding
the cubic lattice in a Hilbert space and characterizing the
automorphisms of the operation $U_\De$. As we have seen the cubic
lattice can be viewed as a semantic representation of the building
blocks of quantum stabilizer codes. An error correction code is only
as useful as the errors it prevents, this section will show that our
characterization of the automorphisms will also characterize these errors.

\begin{defin}
  Given an $[n,k]$ stabilizer code, $C(S)$, $S = \lan g_1,\dots, g_{n -
  k} \ran$. We define $N_{G_{|I|}}(S)$ as the normalizer of $S$ in
  the Pauli group, and $C_{G_{|I|}}(S)$ as the centralizer of $S$ in
  the Pauli group.
\end{defin}

Stated simply the goal of quantum error correction code code is the
following. Given some basis of computation, in our case, a subspace
$V \subseteq \Co^{2^{|I|}}$, we want to prevent unanticipated unitary
actions on $V$.

The spirit of stabilizer codes is start with an abelian subgroup, $H$,
generated by elements in $G_{|I|}$ and define $V = Stab(H)$.
\begin{prop}\label{error types}
  As outlined in \cite[Subsection 10.5.5]{nielsen}, we can then
  separate errors $E \in G_{|I|}$, into three categories. For a given $E$,
  \begin{enumerate}
    \item $E \notin C_{G_{|I|}}(S)$
    \item $E \in S$
    \item $E \in C_{G_{|I|}}(S) - S$
  \end{enumerate}
\end{prop}

However, we only need to worry about the last category as the first
two are correctable.

\begin{theorem}\label{Error-correction condition for stabilizer codes}
  Let $S$ be the stabilizers of an error correction code $C(S)$.
  Suppose that $\{E_j\}$ is
  set of operators in $G_n$ such that $E_j^\ast E_k \notin
  N_{G_{|I|}}(S) - S$ for all $j$ and $k$. Then $\{E_j\}$ is a
  correctable set of errors for the code $C(S)$. \cite[Theorem 10.8]{nielsen}
\end{theorem}

As a brief intuition, the first category can be detected by a
projection onto $C(S)^\perp$ being nonzero. As we have not gained any
information about the entangled state, we can correct this without
issue. The second category does not effect the computational basis we
wish to measure, so is of no concern. The final category is not
detectable by the stabilizers and potentially acts non-trivially on
the computational state. In order to detect this, we typically need
to have some additional information about the system. This comes
at the cost of more interaction of the system with the environment,
which potentially leads to even more error.

The focus of this section is the characterization of errors $E \in
C_{U_{2^{|I|}}}(S) - S$ which we will denote by $\mathscr{E}$.  We have a
couple of immediate questions about the nature of $\mathscr{E}$.
Here, $U_{2^{|I|}}$ is the unitary group of $B(\Co^{2^{|I|}})$. We frame
the results from our prior work in the more standard nomenclature of
the previous section to answer the following questions.
\begin{enumerate}
  \item Errors are represented as unitary operators acting on states,
    so do we need to be concerned with $C_{U_{2^{|I|}}}(S) - S$ instead
    of $C_{G_{|I|}}(S) - S$?
  \item $\mathscr{E} \cup Z(S)$ are the unitary operators contained
    in $S'$ so can we give a nice characterization of this group?
\end{enumerate}


\begin{defin}
  The hamming distance, $d: G_n^+ \r
  \mathbb{N}$, is defined as $d_h(g,h) = |\{i \le n \, | \, g_i \ne
  h_i\}|$ for any $g, h \in G_n^+$.
\end{defin}

The first question can be succinctly answered.

\begin{theorem}
  For any stabilizer code $C(S)$,  $W^\ast(C_{G^+_{|I|}}(S)) =
  W^\ast(C_{U_{2^{|I|}}}(S))$.
\end{theorem}

\begin{proof}

  Since $G^+_{|I|} \subseteq  U_{2^{|I|}}$,
  $C_{G^+_{|I|}}(S) \subseteq  C_{U_{2^{|I|}}}(S)$. Therefore, the
  first inequality follows,
  $rank(W^\ast(C_{G^+_{|I|}}(S))) \le
  rank(W^\ast(C_{U_{2^{|I|}}}(S)))$.

  The reverse inequality follows by a counting argument.  We use that
  $rank(V_S) * |S| = 2^n$, $|S| = 2^{n-k}$, and $rank(V_S) = 2^k$. By
  the maximality of $S$, we have that $rank(W^\ast(C_{G^+_{|I|}}(S)))
  =  rank(W^\ast(S))$. As the elements of $G_{|I|}^+$ are orthogonal under
  the Hilbert Schmidt inner product, $rank(W^\ast(S)) = |S| = 2^{n-k}$.

  $|G_{|I|}^+| = 4^n$, so the positive Pauli matrices also form an
  orthogonal basis of $M_{2^n}(\mathbb{C})$.
  Suppose $U \in W^\ast(C_{U_{2^{|I|}}}(S))$ such that $Ug = gU$ and $g \in
  |G_{|I|}^+| - S$. First we will order the
  elements of $G_{|I|}^+$, so the first $2^{n-k}$ elements are in $S$, and we
  assume without loss of generality that $g = g_{2^{n-k} + 1}$.
  Then $U = \sum_{i=1}^{2^{n-k}+ 1} c_i g_i$, where $c_{2^{n-k} + 1}
  \ne 0$. Then  $g_{2^{n-k} + 1} \in C_{U_{2^{|I|}}}(S) \cap G_{|I|}^+$, so
  $g_{2^{n-k} + 1} \in C_{G^+_{|I|}}(S) = S$, which contradicts that
  maximality of $S$. Therefore for all  $U \in W^\ast(C_{U_{2^{|I|}}}(S))$,
  $U = \sum_{i=1}^{2^{n-k}} c_i g_i$, so
  $rank(W^\ast(C_{U_{2^{|I|}}}(S))) \le 2^{n-k}$. In summation,
  $rank(W^\ast(S)) = |S| = 2^{n-k} \le rank(W^\ast(C_{U_{2^{|I|}}}(S)))
  \le 2^{n-k}$, and $rank(W^\ast(S)) =
  rank(W^\ast(C_{U_{2^{|I|}}}(S)))$. We conclude that
  $W^\ast(C_{G^+_{|I|}}(S)) = W^\ast(C_{U_{2^{|I|}}}(S)$ since
    $W^\ast(C_{G^+_{|I|}}(S))$ is a subspace of
    $W^\ast(C_{U_{2^{|I|}}}(S))$ of equal rank.
  \end{proof}






  \begin{defin}
    Let $H$ be the Hadamard matrix.
  \end{defin}

  \begin{defin}
    For $g \in \{G_1^+, H\}$, let $g^{\otimes I} = \otimes_{i \in I}
    g_i$, $g_i = g$ for all $i \in I$.
  \end{defin}

  We first introduce a standard result.
  \begin{prop}\label{commutant of tensor products}
    Let $U \in B(H)$ and $V \in B(K)$ be distinct unitary operators of
    finite order. Then $(W^\ast(U \otimes I_K)\otimes W^\ast(I_H
    \otimes V))' = W^\ast(W^\ast(U \otimes I_K)' \otimes W^\ast(I_H
    \otimes V)')$ \cite{commutant-of-tensor-products}.
  \end{prop}

  We begin with some assumptions that are standard in any reasonably
  effective stabilizer error correction scheme.
  \begin{defin}
    Let a standard stabilizer be a stabilizer code with the following
    additional conditions.
    \begin{enumerate}
      \item $|S| \ge 2$
      \item The code must attempt to correct both bit flips (detected
        by $X$) and phase flips (detected by $Z$).
    \end{enumerate}
  \end{defin}

  \begin{exmp}
    To demonstrate that definition of a standard code is reasonable, we
    note that the famous typical examples of stabilizer codes fit
    this pattern.
    \begin{enumerate}
      \item  Shor code: $S = \lan ZZIIIIIII, ZIZIIIIII, IIIZZIIII,
        IIIZIZIII, IIIIIIZZI, IIIIIIZIZ, \\ XXXXXXIII, XXXIIIXXX \ran$
      \item Stene code: $S = \lan IIIXXXX,IIIZZZZ, IXXIIXX, IZZIIZZ,
        XIXIXIX, ZIZIZIZ \ran$
      \item  Five qubit code: $S = \lan XZZXI, IXZZX, XIXZZ, ZXIXZ \ran$

        In fact, the five qubit code is the minimal size of a
        stabilizer code encoding a single qubit, which corrects
        arbitrary errors on a single qubit.
    \end{enumerate}
  \end{exmp}



  \begin{prop}\label{coxeter representation}
    Let $g \in G_n^+,$ $d_h(g, I) = n$,  then there exists a unitary
    representation $\tau_{\rho(g)} : B_n \r M_{2^n}(\Co)$ such that
    $\{1,\rho(g)\} = Z(\tau_{\rho(g)}(B_n))(\Co)$.
  \end{prop}

  \begin{proof}
    By unitary similarity, we can take $\rho(g) = U_\De$ and let
    $\tau_{\rho(g)}$ be defined in \cite[Lemma
    5.1.13]{turnansky_thesis}, and, lastly, we use that for $|S| = n,
    Per_{\De}(C) \cong B_n$.
  \end{proof}

  \begin{lemma}\label{isomorphism centralizer of single element}
    Let $g \in G_n^+,$ $d_h(g, I) = n$, and $\tau_{\rho(g)} : B_n \r
    M_{2^n}(\Co)$ be as constructed in Proposition \ref{coxeter
    representation}. Then for any unitary $U \in \rho(g)'$, there exists
    a unitary $V \in \tau_{\rho(g)}(B_n)$ and $h \in G_n^+$ such that
    $U = V\rho(h)$, where $g,h$ commute.
  \end{lemma}

  \begin{proof}
    By unitary similarity we can take $\rho(g) = U_\De$, and, there
    exists an $h \in G_n^+$ such that $U = V \rho(h)$ and $g,h$ commute
    by \cite[Lemma 5.1.16]{turnansky_thesis}.
  \end{proof}

  \begin{theorem}\label{isomorphism centralizer of two elements}
    Let $C(S)$ over be a standard stabilizer code on $2k$ qubits with
    generators $\{g, H^{\otimes n}g(H^{\otimes n})^\ast\}$ such that
    $d_h(g, I) = 2k$, then $C_{U_{2^k}}(S) = (K \times H^{\otimes
    n}K(H^{\otimes n})^\ast)\rtimes S_n$, where $K \trianglelefteq \Z_2^{2k}$.
  \end{theorem}

  \begin{proof}
    The result follows by Lemma \ref{isomorphism centralizer of single
    element} when we consider the additional constraint that
    $\tau_{\rho(g)}(B_n)$ must commute with $h$. By construction, these
    are the elements in $K \le \tau_{\rho(g)}(B_n)$ generated by $
    \{U = \otimes_{i = 1}^{2k} U_i\,| U_i \in G_{2k}^+\text{ and }
    d_h(U, g) \, \% \, 2 = 0\}$.

    By symmetry we also get a representation  $\tau_{\rho(h)}(B_n)$. In
    either case, $S_n \le B_n$ acts by permuting the indices of the
    tensor product. Then we have $H^{\otimes n}K(H^{\otimes n})^\ast
    \le W^\ast(g)' \cap W^\ast(h)'$, and for any $\alpha \in H^{\otimes
    n}K(H^{\otimes n})^\ast$ and $\sigma \in S_n$, we can consider
    $\sigma^{-1}\alpha\sigma$ to get our desired representation.
  \end{proof}

  \begin{lemma}\label{mutual commutant base case}
    Given the above notation, Let $C(S)$ be a standard stabilizer code
    on $2k$ qubits with generators and $\{g, H^{\otimes n}g(H^{\otimes
    n})^\ast\}$ such that $d_h(g, I) = 2l$ for some $l < k$, then
    $C_{U_{n}}(S) =  ((K \times H^{\otimes n}K(H^{\otimes
    n})^\ast)\rtimes S_l) \otimes M_{2^{k - l}}(\Co)$, where $K
    \trianglelefteq \Z_2^{l}$.
  \end{lemma}

  \begin{proof}
    This follows from Proposition \ref{commutant of tensor products}
    and Theorem \ref{isomorphism centralizer of two elements}.
  \end{proof}

  We can now see that in the finite case, the centralizer of $2$ elements
  is essentially an intersection of different representations of the
  group $B_n$.

  \begin{theorem}
    Let $C(S)$ be a standard stabilizer code
    on $n = 2k$ qubits with generators set $\{g\}_{i=1}^l$, $l \le k$.
    Then $C_{U_{2^k}}(S) = \cap_{i=1}^n C_{U_{2^k}}(g_i)$. In addition,
    $W^{\ast}(C_{U_{2^k}}(S))$ is generated by a representation of $B_j$
    for some $j \le 2k$.
  \end{theorem}

  \begin{proof}
    The result is an immediate consequence of Lemma \ref{mutual
    commutant base case}, and that the centralizer of a set of elements
    is the intersection of the centralizer for each respective element.
  \end{proof}

  As we have already shown
  that the von Neumann algebra generated by the centralizer includes
  all unitary operators that may result in non trivial errors on a
  quantum system, we now have a strong guarantee of errors on a
  quantum system. By
  preventing errors that may be represented as the Coxeter group, we in
  fact can control for all non trivial errors.



  We want to highlight that we have a guarantee stronger than the
  original claim even in the finite case where the errors for a given
  stabilizer group are assumed to be in the Pauli group, and we have a
  generalization to all non-trivial errors in the unitary group.

  One may now realize that the nature of cubic lattices and quantum
  stabilizer codes are intertwined. Specifically, we have shown that
  the automorphisms of the cubic lattice are highly related to the
  non-trivially correctable errors of a quantum stabilizer code. In
  fact this realization may be an opportunity. We now have a propositional
  logic which does not negatively interact with the stabilizer of a
  quantum code. Perhaps this will allow us to consider well structured logical
  implications on an entangled state.

\appendix
\section{Cubic Lattices as a fault tolerant measuring system}

By the results of the previous sections, any algorithm whose logic is
consistent with the three valued Lukasiewicz logic can be implemented
on a qubit based computational system. As an example we include the
Renyi-Ulam game.

The Renyi-Ulam game and its generalizations \cite{history_of_liars}
can be viewed as an interactive proof. The verifier wants to guess a
pre-selected number from the range $0, \dots, n$ and is allowed to
ask the prover questions until sufficiently convinced of the
pre-selected number. We will be interested in the original Renyi-Ulam
game where the prover is allowed to lie at most once. An optimal
strategy is known and is presented in \cite{searching_with_a_lie}. In
\cite{paraconsistency}, it is stated that the states of the
Renyi-Ulam game are exactly equivalent to the states of the cubic
lattice. As we have a quantum gate realization the cubic lattice, we
can give a quantum gate based realization of this algorithm
consisting of only Clifford gates and additions to turn the
symmetries into projections.

We separate the roles of the verifier and the prover. The prover can
be a quantum
oracle and embedded into any fixed $[n,k]$ stabilizer code, and the
verifier can run on a classical machine. In this instance, one can
view the verifier as measuring the hidden entangled state of
the quantum system. Importantly the measurements, i.e. projection
operators, commute with the stabilizers of the error correction code,
so this game will not conflict with the desired stabilized state of
the system. Lastly, as an immediate application, we can recover from
faulty measurement and obtain recover the correct value.\\

For ease of the reader, we include the following python implementation of
\cite{searching_with_a_lie}.
\begin{lstlisting}[language=Python]
class Prover:
    """
    A prover which returns an answer from the Verifier.
    The response can be true or a lie decided at
    random, but at most one lie can be told throughout the game.

    Attributes:
        hidden_value: The value the verifier is trying to guess.
        honest: If the prover should lie or fully honest.
        n: The number of possible values [0,1 , .. n].
    """

    def __init__(
        self,
        n: int,
        hidden_value: int,
        honest: bool = True,
    ) -> None:
        self.hidden_value = hidden_value
        self.n = n
        self.has_lied = False
        self.honest = honest

    def honest_answer(self, subset: SortedSet) -> bool:
        """
        Given a subset, respond if the hidden value is in the subset.

        Args:
            The subset supplied by the Verifier.

        Returns:
            The true answer to the question.
        """
        if self.hidden_value in subset:
            return True
        return False

    def answer(self, subset: SortedSet) -> bool:
        """
          Given a subset, respond if the hidden value is in the
          subset or lie randomly.

        Args:
            The subset supplied by the Verifier.

        Returns:
            The answer to the question.
        """
        result = self.honest_answer(subset)
        if self.honest or self.has_lied:
            return result
        else:
            lie = bool(random.getrandbits(1))
            if lie:
                self.has_lied = True
                return not result
            return result


class Verifier:
    """
    The Verifier must guess the hidden value of the prover.

    Attributes:
        prover: The prover it is querying.
        truth: The possible true values at the current iteration.
        lie: The values at the current iteration that may be true
        if the verifier lied.
    """

    def __init__(self, prover: Prover) -> None:
        self.prover = prover
        self.truth = SortedSet(range(1, prover.n + 1))
        self.lie = SortedSet()

    def ask(
        self,
        truth_subset: SortedSet,
        lie_subset: SortedSet,
    ) -> bool:
        """
        Ask the Prover if the hidden value is in the truth subset or
        the lie subset.

        Arguments:
            truth_subset: A potential truth subset to be queried.
            lie_subset: A potential subset to be queried for value
            that may be true if the prover has lied.

        Returns:
            True if the hidden value is either subset.
        """
        if self.prover.answer(subset=truth_subset) or self.prover.answer(
            subset=lie_subset
        ):
            return True
        return False

    def compute_char(
        self,
        truth_size: int,
        lie_size: int,
    ) -> int:
        """
        Helper function to compute the current char, A helper variable
        to decide the new cutoff point.

        Arguments:
            truth_size: The size of the current truth set.
            lie_size: The size of the current lie set.
        Returns:
            The new char value.
        """
        j = 1
        while truth_size * (j + 1) + lie_size > 2**j:
            j += 1
        return j

    def result(
        self,
        answer: bool,
        truth_cut: int,
        lie_cut: int,
        char,
    ) -> tuple[int, int, int]:
        """
        Given the response from the Prover, update variables for the next
        iteration.

        Arguments:
            answer: The response from the prover.
            truth_cut: The cutoff point for truth set
            lie_cut: The cutoff point for the lie set
            char: A helper variable to decide the new cutoff point.
        Returns:
            The updated values for size of truth set, size of lie set,
            and char.
        """
        if answer is True:
            self.lie = SortedSet(self.lie.islice(0, lie_cut))
            self.lie.update(self.truth.islice(truth_cut))
            self.truth = SortedSet(self.truth.islice(0, truth_cut))
        else:
            self.lie = SortedSet(self.lie.islice(lie_cut)).update(
                self.truth.islice(0, truth_cut)
            )
            self.truth = SortedSet(self.truth.islice(truth_cut))

        char = self.compute_char(len(self.truth), len(self.lie))
        return [len(self.truth), len(self.lie), char]

    def compute_ch(
        self,
        truth_size: int,
    ) -> int:
        """
        Update ch. A helper variable to decide new cutoff.

        Arguments:
            truth_size: The size of the current truth set.

        """
        ch: int = 0
        while truth_size * (ch + 1) + (truth_size - 1) > 2**ch:
            ch += 1
        return ch

    def strategy(self) -> int:
        """
        The primary function to orchestrate the guessing strategy.

        Returns:
            The hidden value.
        """

        # Initialize stragegy variables.
        truth_size = len(self.truth)
        lie_size = len(self.lie)
        char = 0
        answer = self.ask(
            SortedSet(self.truth.islice(0, (truth_size // 2))),
            self.lie,
        )
        truth_size, lie_size, char = self.result(
            answer=answer,
            truth_cut=truth_size // 2,
            lie_cut=0,
            char=char,
        )
        ch = self.compute_ch(truth_size=truth_size)

        # Primary case to reduce to several end cases.
        while truth_size > 1:
            if truth_size % 2 == 0:
                t: int = truth_size // 2
                l: int = truth_size // 2
            else:
                t = (truth_size + 1) // 2
                l = max((truth_size - 1) // 2 - ch // 2, 0)

            weight = (t * char) + l + truth_size - t
            adjust = min(
                lie_size - (truth_size - 1),
                2 ** (char - 1) - weight,
            )
            l += adjust
            answer = self.ask(
                truth_subset=SortedSet(self.truth.islice(0, t)),
                lie_subset=SortedSet(self.lie.islice(0, l)),
            )
            truth_size, lie_size, char = self.result(
                answer=answer,
                truth_cut=t,
                lie_cut=l,
                char=char,
            )
            ch = self.compute_ch(truth_size)

        # Either the hidden value is only remaining value in the
        # truth set or it is in the lie set.
        # Reduce the lie set until it can be directly deduced.
        while truth_size == 1 and lie_size >= char:
            answer = self.ask(
                truth_subset=self.truth,
                lie_subset=SortedSet(
                    self.lie.islice(0, (lie_size - char + 1) // 2),
                ),
            )
            truth_size, lie_size, char = self.result(
                answer=answer,
                truth_cut=truth_size,
                lie_cut=(lie_size - char + 1) // 2,
                char=char,
            )

        #  Either the hidden value is only remaining value in the truth set
        #  or it is in the lie set.
        # Determine if the value is in the lie set.
        if truth_size == 1 and lie_size > 0:
            answer = self.ask(self.truth, SortedSet())
            truth_size, lie_size, char = self.result(
                answer=answer,
                truth_cut=truth_size,
                lie_cut=0,
                char=char,
            )
        # Lie set has been eliminated, so the hidden value is in the
        # truth set.
        if truth_size == 1:
            return self.truth.pop()

        # Truth set has been eliminated, so the hidden value is in the
        # lie set.
        while lie_size > 1:
            answer = self.ask(
                SortedSet(),
                SortedSet(self.lie.islice(0, lie_size // 2)),
            )
            truth_size, lie_size, char = self.result(
                answer=answer,
                truth_cut=0,
                lie_cut=lie_size // 2,
                char=char,
            )

        # Truth set is empty and lie set has one element, so it must be the
        # hidden element.
        return self.lie.pop()
\end{lstlisting}


\printbibliography[
heading=bibintoc,
title={References}
]
\end{document}